\documentclass[copyright,creativecommons]{eptcs}
\providecommand{\event}{TARK 2025} 

\usepackage{iftex}

\ifpdf
  \usepackage{underscore}         
  \usepackage[T1]{fontenc}        
\else
  \usepackage{breakurl}           
\fi
\usepackage{tikz}
\usetikzlibrary{shapes,arrows,positioning,calc,shapes} 
\usepackage{float}  
\usepackage{longtable}
\usepackage{multirow}
\usepackage{amssymb}
\usepackage{amsmath}
\usepackage{pifont}
\usepackage{stmaryrd}
\usepackage{comment}
\usepackage{bussproofs}
\usepackage[all, knot]{xy}
\usepackage{color}
\usepackage[english]{babel}
\usepackage{comment}
\usepackage{xcolor}
\usepackage[ruled,vlined]{algorithm2e}
\usepackage{caption}
\newenvironment{proof} {\textsc{Proof}\quad} {\hfill $\Box$\\}
\definecolor{citecolor}{rgb}{0.5,0.5,0.5}

\newcommand{\M}{\mathcal{M}}

\newcommand{\la}{\langle}
\newcommand{\ra}{\rangle}
\newcommand{\rel}[1]{\xrightarrow{#1}}

\newcommand{\lr}[1]{\langle #1 \rangle}
\newcommand{\llrr}[1]{\llbracket #1 \rrbracket}

\newcommand{\DKH}{\mathrm{\mathbf{DKH}}}

\renewcommand{\phi}{\varphi}

\newtheorem{theorem}{Theorem}
\newtheorem{definition}[theorem]{Definition}
\newtheorem{lemma}[theorem]{Lemma}
\newtheorem{proposition}[theorem]{Proposition}
\newtheorem{example}[theorem]{Example}
\newtheorem{corollary}[theorem]{Corollary}
\newtheorem{claim}{Claim}[theorem]

\newenvironment{claimproof}{\par\noindent\textit{Proof of Claim \theclaim:}\space}{\hfill $\blacksquare$}

\newcommand{\Kh}{\mathsf{Kh}}

\newcommand{\K}{\mathsf{K}}

\renewcommand{\phi}{\varphi}

\DeclareSymbolFont{symbolsC}{U}{txsyc}{m}{n}
\DeclareMathSymbol{\strictif}{\mathrel}{symbolsC}{74}

\newcommand{\SDKh}{\mathbb{SDKH}}

\newcommand{\CELeafN}{\ensuremath{\mathtt{CELeaf}}}
\newcommand{\CEInnerN}{\ensuremath{\mathtt{CEInner}}}
\newcommand{\Dom}{\ensuremath{\mathtt{dom}}}

\title{Distributed Knowing How}
\author{Bin Liu
\institute{Department of Philosophy, Peking University}
\email{liub85@stu.pku.edu.cn}
\and
Yanjing Wang
\institute{Department of Philosophy, Peking University}
\email{y.wang@pku.edu.cn}
}

\def\titlerunning{Distributed Knowing How}
\def\authorrunning{Liu \& Wang}

\begin{document}
\maketitle

\begin{abstract}

Distributed knowledge is a key concept in the standard epistemic logic of knowledge-that. In this paper, we propose a corresponding notion of distributed knowledge-how and study its logic. Our framework generalizes two existing traditions in the logic of know-how: the individual-based multi-step framework and the coalition-based single-step framework. In particular, we assume a group can accomplish more than what its individuals can jointly do. The distributed knowledge-how is based on the distributed knowledge-that of a group whose multi-step strategies derive from distributed actions that subgroups can collectively perform. As the main result, we obtain a sound and strongly complete proof system for our logic of distributed knowledge-how, which closely resembles the logic of distributed knowledge-that in both the axioms and the proof method of completeness.

\end{abstract}

\section{Introduction}
One charming feature of epistemic logic is the formalization of group knowledge, such as common knowledge and distributed knowledge \cite{RAK}. In particular, distributed knowledge plays an important role in various applications of epistemic logic, which also brings technical challenges that go beyond the bisimulation-invariant modal logic (cf. e.g., \cite{AgotnesW17}). 

While epistemic logic mainly focuses on propositional \textit{de dicto} knowledge (knowledge-that), there is an increasing interest in studying \textit{de re} knowledge, such as knowledge-how/why/what, and so on \cite{sep-logic-epistemic}. It is natural to ask what the suitable notions of group knowledge-wh are \cite{Wang2018bkt}. In \cite{su2017distributed,fan2021commonly}, various common/distributed notions of \textit{knowing whether} have been proposed and studied. In this work, we try to propose a notion of distributed \textit{knowing how} and study its logic. 

We build our work by first combining two logical frameworks of knowing how. Initiated in \cite{Wang15:lori,wang2018goal}, the planning-based multi-step framework was inspired by the philosophical discussion of knowledge-how and the theory of automated planning in AI. In such a logic, one knows how to achieve $\phi$ means there is a \textit{plan} consisting of her own actions such that the agent knows that the plan can be executed and will always terminate and guarantee $\phi$. The plans can be linear, branching, or even a program written in a programming language \cite{LiW21}. The second framework is a coalition-based one-step setting initiated in \cite{naumov2017together,naumov2018together}, where a coalition knows how to achieve $\phi$ means there is a \textit{joint action} of the coalition such that the coalition knows that this joint action can make sure $\phi$ no matter what others do. To some extent, the second framework captures a simple form of distributed knowledge-how as \cite{LiW19} remarked, but in this work we aim for a much more general notion which requires the power of both frameworks.  The differences between these two existing approaches and our new ``mixed'' approach can be summarized below: 
\begin{center}
\begin{tabular}{|c|c|c|c|}
  approach  & know-that & know-how & actions  \\
  \hline
   planning-based &  individual & multi-step  & individual \\
   coalition-based & distributed & one-step &  joint\\
   mixed & distributed & multi-step & group 
\end{tabular}    
\end{center}
The last column specifies what actions are used. The first approach assumes that the actions by individuals can already move the states. The second approach presupposes that the transitions only happen due to the joint actions of \textit{all} the agents, which is suitable for game-theoretical settings. In contrast, actions in our approach are somehow the generalization of the two, featuring what we call \textit{group actions} that may or may not be decomposable into individual actions. To see the need for our mixed approach, let us see a few examples of knowledge-how intuitively distributed in a group based on various types of group actions, such that only the group knows how to achieve a certain goal without each member knowing how to do it. 

The first example is about actions that agents can only do together as a group. 
\begin{example}[Couch moving]\label{ex.extra}
No one can single-handedly move a heavy couch through a narrow doorway around a corner. However, if two people work together, they know how to do it through a synchronized group action of one person lifting and maneuvering one end of the couch and the other handling the opposite end of the couch. 
\end{example}
The group action here in the example belongs to the group as a whole and is not decomposable into \textit{independent} individual actions. We will take such actions as \textit{atomic} in the framework. In contrast with the joint actions in the coalition-based game-like framework (and ATL-like epistemic frameworks \cite{vanderHoek2003,JamrogaA07}), other agents outside the group do not affect the result of such a group action, which is similar to the individual actions in the planning-based approach. 

Moreover, distributed knowledge-how arise when a group can use actions from its members or subgroups.
\begin{example}[Theorem proving]\label{ex.col}
In collaborations on mathematical work, it often happens that two researchers do not know how to prove a theorem independently, but together they do. For example, $i$ knows how to prove a lemma, and $j$ knows how to prove a theorem given the lemma is true, then together they should know how to prove the theorem by first proving the lemma using $i$'s knowledge-how and then proving the theorem using $j$'s knowledge-how. Intuitively, the knowledge-how to prove the theorem is (sequentially) distributed in the group of these two researchers. 
\end{example}


This example shows that the distributed knowledge-how of a group may be based on a plan making use of actions inherited from its members or  subgroups. 


Yet another source of distributed knowledge-how is based on actions of a group that are \textit{decomposable} into \textit{independent} actions of its members and subgroups. It resembles the joint actions in the coalition-based approach, but not necessarily for the full set of agents. 

\medskip

\hspace{-1.5em}\begin{minipage}{0.7\linewidth}
\begin{example}[Joint treatment]\label{ex.joint}
Suppose a patient has two health problems represented by $p$ and $q$. Doctor $1$ knows how to cure $p$ by treatment $a$ but is not sure whether it can cure $q$, and Doctor $2$ knows how to treat $q$ by treatment $b$ but is not sure whether it can cure $p$. Then, assuming the treatments are independent (not affecting each other), by doing both $a$ and $b$ respectively, the doctors distributedly know how to cure both problems (making sure $\neg p\land\neg q$).
\end{example}
\end{minipage}
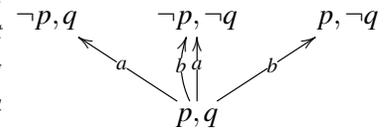
\begin{minipage}{0.34\linewidth}
$\xymatrix{
\neg p, q & \neg p, \neg q & p, \neg q \\
                   & p,q\ar[lu]|a\ar[u]|a\ar@/^/[u]|b \ar[ru]|b\\  
}$
\end{minipage}
\medskip

Compared to the group action in Example \ref{ex.extra}, the effect of the joint action in Example \ref{ex.joint} is determined and decomposable by the effects of the independent composing actions. As we will see later, the set of outcomes induced by such a joint action is the \textit{intersection} of the outcomes caused by each composing action.

From these three examples, we can summarize the following observations respectively: 
\begin{itemize}
\item[O1] A group may have some \textit{extra} irreducible actions that none of its proper subgroups can do.
\item[O2] The available actions of a group at least include all the actions of the subgroups. 
\item[O3] The group can perform decomposable \textit{joint actions}, whose outcomes are the intersection of the outcomes of the composing actions (by subgroups).
\end{itemize}
Thus, there are three types of group actions relevant for distributed knowledge-how: (I) irreducible group actions belonging to the whole group but not less; (II) actions inherited from the subgroups; (III) the joint actions computed as the ``intersections'' of the composing actions.   

We will formalize the above idea by first giving ``atomic'' group actions $A_G$ of type (I) for each group $G$ (including singleton groups as individuals). By O1, if $H\subsetneq G$ then $A_H\cap A_G=\emptyset.$ Next, we compute the closure $A^*_G$ (call it the set of distributed actions) recursively with respect to all the inherited actions of type (II) and the joint actions of type (III). In particular, the joint actions of group $G$ are composed of actions belonging to any non-singleton partial partition of the group, not just a tuple of actions from each member. This reflects the fact that a task force may be divided into subteams with their own know-how. We assume that for joint actions, each person can remain idle but cannot be used twice \textit{at the same time}\footnote{Otherwise, an agent may have the unrealistic ability to execute two actions at the same time. For example, if we assume $a$ and $b$ in Example \ref{ex.joint} are both actions for Doctor 1, then Doctor 1 can team up with herself, to execute $\lr{a,b}$ to guarantee the outcome $\neg p \land\neg q$, which may not be possible.}, thus requiring the partial partition (partition of a subset). 

Below is an example to illustrate the definition of group actions, where we write $A_{\{i\}}$ as $A_i$ for brevity and present the joint actions according to a fixed ordering of agents and groups. 
\begin{example}\label{ex.Astar}
Suppose there are two agents $1,2$, the atomic group actions of type (I) are defined as $A_{1}=\{a,b\}$, $A_{2}=\{c\}$,  $A_{\{1,2\}}=\{d\}$. Then the closures $A^*_G$ are computed as: $A^*_i=A_i$ for singleton sets $\{i\}$; $A^*_{\{1,2 \}}=A_{\{1,2\}}\cup (A_{1}\cup A_{2})\cup (A^*_{1}\times A^*_{2})=\{d, a,b,c, \lr{a,c}, \lr{b,c}\}$ where $A_{\{1,2\}}$ is the set of type (I) group actions, and $A_{1}\cup A_{2}$ is the set of inherited actions of type (II), and $(A^*_{1}\times A^*_{2})$ is the set of intersecting joint actions of type (III). Note that if there are more than two agents, then the actions of type (III) become complicated, e.g., $A^*_{\{1,2,3\}}$ should include $A^*_{\{1,2\}}\times A^*_{3}$, $A^*_{1}\times A^*_{\{2,3\}}$, $A^*_{\{1,3\}}\times A^*_{2}$, $A^*_{1}\times A^*_{2}$, $A^*_{2}\times A^*_{3}$, $A^*_{1}\times A^*_{3}$, and $A^*_1\times A^*_2\times A^*_3$, i.e., it includes the joint actions of any non-singleton partial partition of the group representing how the work can be distributed among subgroups.

\end{example}

Distributed knowledge-how is then based on multi-step strategies that make use of actions in $A^*_{G}$ and distributed knowledge-that. We say a (possibly singleton) group $G$ (distributedly) knows how to achieve $\phi$ if there is a strategy $\sigma$ making use of $A^*_G$ such that it is distributed knowledge of $G$ that $\sigma$ will terminate and guarantee $\phi$. Consider the following example extending Example~\ref{ex.joint}:
\begin{example}[Distributed knowledge-that and -how]
    On top of Example~\ref{ex.joint}, we further assume that the two doctors have some uncertainty regarding the patient's problems: Doctor 1 is sure the patient has $p$ but not sure about $q$, and Doctor 2 is sure the patient has $q$ but not sure about $p$.  Furthermore, if $p\land \neg q$ then Doctor 1 should do treatment $c$ (not $a$) to cure $p$, and if $\neg p\land  q$ Doctor 2 should do treatment $d$ (not $b$) to cure $q$.  The situation can be formalized in the following model where $A_1=\{a,c\}$, $A_2=\{b, d\}$ and $A_{\{1,2\}}=\emptyset$. The dotted lines represent the epistemic indistinguishability relations of the agents, and we omit the reflexive arrows.
$$
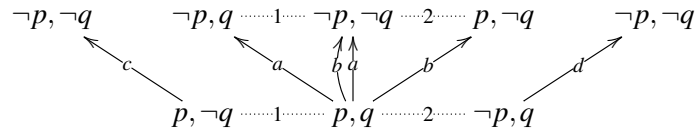
\xymatrix{
  \neg p,\neg q & \neg p, q \ar@{.}[r]|1& \neg p, \neg q \ar@{.}[r]|2& p, \neg q & \neg p, \neg q \\
                   & p, \neg q\ar[ul]|c & p,q\ar@{.}[l]|1\ar@{.}[r]|2\ar[lu]|a\ar[u]|a\ar@/^/[u]|b \ar[ru]|b&  \neg p, q\ar[ur]|d\\  
}$$
Intuitively, at world $(p,q)$, Doctor 1 does not know how to achieve $\neg p$ as she does not have a uniform move to do so over the indistinguishable worlds $(p,\neg q)$ and $(p,q)$. Similarly, Doctor 2 does not know how to achieve $\neg q$. However, the group of the two doctors $1$ and $2$ distributedly knows how to make sure $\neg p\land\neg q$ by doing the joint action $\lr{a,b}\in A^*_{\{1,2\}}$ based on their distributed knowledge, which leaves them with a sole possible world $(p, q)$.   
\end{example}


In this paper, we propose such a mixed framework of distributed knowing how, based on the group actions explained above and the notion of strategies proposed in \cite{fervari2017strategically}. The main technical contribution is a sound and complete proof system of the logic with intuitive axioms, whose completeness proof is based on the unraveling of both the epistemic relations and the action transitions. One crucial difference from the systems of coalition-based know-how logics is that the \textit{cooperation axiom} that plays a key role there is no longer valid here due to the multi-step setting. 
In the following, Section \ref{sec.lands} lays out the language and semantics for our logic of distributed knowledge-how. Section \ref{sec.axiom} gives a proof system. In Section \ref{sec.comp}, we prove the completeness. Finally, we conclude and discuss future work in Section \ref{sec.conc}.

\section{Language and semantics} \label{sec.lands}

We present the following syntax and semantics given the set $P$ of proposition letters and the finite set $I=\{i_0,\cdots,i_\mathfrak{n}\}$ of agents.\footnote{The requirement for the finiteness of the set $I$ of agents is solely for the convenience of denoting distributed actions and has no effect on the logic (valid formulas).}

\begin{definition}[Language]\label{defLang}
    The language $\DKH$ is defined by the following BNF where $p\in P$ and $G$ is a subset of $I$ called group:
    \[\phi::=\top\mid p\mid\neg\phi\mid (\phi\land\phi)\mid \K_G\phi \mid \Kh_G\phi \]
\end{definition}

\begin{definition}[Model]\label{defmodel}
    A model $\M$ is a tuple $\left\la S, \{\sim_i\}_{i\in I}, \{A_G\}_{G\subseteq I}, \{\rel{a}\mid a\in\bigcup_{G\subseteq I}A_G\}, V\right\ra$, where
  \begin{itemize}
  \item $S$ is a set of states;
  \item $\sim_i$ is an equivalence relation on $S$ for each $i\in I$;
  \item $A_G$ is a set of atomic group actions for each $G\subseteq I$ such that: $G\subsetneq H$ implies $A_G\cap A_H=\emptyset$;  $A_\emptyset=\emptyset$;
  \item $\rel{a}$ is a binary relation on $S$ for each $a\in\bigcup_{G\subseteq I}A_G$, called the transition relation of $a$;
  \item $V:P\to\mathcal{P}(S)$ is a valuation function.
  \end{itemize}
\end{definition}

We use $A^+_{ G}$ to denote $\bigcup_{G'\subseteq G}A_{G'}$ for any $G\subseteq I$, thus $A_I^+$ is the set of all atomic group actions and $G\subseteq H$ implies $A^+_{ G}\subseteq A^+_{ H}$. Note that the actions only appear in the models, but not in the language, and the transitions can be non-deterministic. The models in this paper closely align with those in \cite{LiW24} with some notable distinctions. First, instead of the set of actions of each agent, we have the set of actions for each group. Moreover, as mentioned in the introduction, we will compute the closures of these atomic action sets, which serve as the basis for distributed knowledge-how. Second, the requirement that $G\subsetneq H$ implies $A_G\cap A_H=\emptyset$, reflects the intuition behind the type (I) group actions that are only available to the whole group and not less.\footnote{Although this requirement is conceptually important, technically, we think it can be dropped without changing the logic, because the closures of action sets are the same with or without the requirement. To tell the difference in logic, we may need to talk about concrete actions in the language.}

\begin{definition}[Distributed indistinguishability]\label{group indistinguishability}
  For each group $G$, $\sim_G$ is a binary relation on $S$ such that $s \sim_G t$ iff $s \sim_i t$ for each $i\in G$.
\end{definition}

It is routine to show that distributed indistinguishability relations are equivalence relations.

\begin{proposition}
    For any group $G$, $\sim_G$ is an equivalence relation on $S$.
\end{proposition}

Therefore, we can define equivalence classes w.r.t. $\sim_G$.

\begin{definition}\label{Equivalence class}
For any group $G$ and $s\in S$, we use $[s]_G$ to denote the equivalence class $\{t\in S\mid s\sim_G t\}$, and use $[S]_G$ to denote the collection of all the equivalence classes on $S$ w.r.t. $\sim_G$.
\end{definition}

Note that by definition $s\sim_\emptyset t$ for any $s,t\in S$, which implies that $[s]_\emptyset=S$ for any $s\in S$.


We will now define the set of \textit{distributed actions} $A^*_G$ as a closure of $A_G$ recursively, which includes the three types of group actions as discussed in the introduction and illustrated in Example \ref{ex.Astar}. To make the technical definition concise, we divide the  actions in $A^*_G$ into two parts: (1) atomic group actions of $G$ and its subgroups, and (2) joint actions in the form of $\lr{d_0,\cdots,d_n}$.  Since the order of $d_0,\cdots,d_n$ is insignificant in our setting, we fix an ordering $\prec$ below on \textit{mutually disjoint groups} to avoid generating redundant joint actions, which will not affect the logic technically. Recall that the set of all agents is $I=\{i_0,\cdots,i_\mathfrak{n}\}$. For any nonempty group $G=\{i_{n_0},\cdots,i_{n_k}\}$, let $\min G=\min\{n_0,\cdots,n_k\}$ and let $G\prec H$ iff $\min G<\min H$. This will always give us a strict ordering over a set of mutually disjoint groups of agents. 
\begin{definition}[Distributed actions]
   For each group $G$, its distributed action set $A_G^*:=A^+_{ G}\cup \\ \{\lr{d_0,\cdots, d_n}\in A_{G_0}^*\times\cdots\times A_{G_n}^*\mid \{G_0,\cdots,G_n\}\mbox{ is a non-trivial partial partition of }G\mbox{ and }G_0\prec\cdots\prec G_n \} $, where a non-trivial partial partition of a set $X$ is a partition of a (not necessarily proper) subset of $X$ such that it is not a singleton set. It follows that $A_\emptyset^*=\emptyset$ and $A_i^*=A_i$ for each $i\in I$.
\end{definition}



Note that joint actions of subgroups of a group $G$ are also inherited because a non-trivial partial partition of a subgroup is also a non-trivial partial partition of $G$. Thus we have the following monotonicity for both distributed indistinguishability and distributed action sets.

\begin{proposition}\label{propMonoAct}
    $G\subseteq H$ implies $\sim_H\subseteq\sim_G$ and $A_G^*\subseteq A_H^*$.
\end{proposition}

\begin{proof}
    Suppose that $G\subseteq H$. $\sim_H\subseteq\sim_G$ is straightforward by definition. We show $A_G^*\subseteq A_H^*$ as follows. If $G=\emptyset$, then $A_G^*=\emptyset$, then we have $A_G^*\subseteq A_H^*$. If $G=\{i\}$ where $i\in I$, then $A_G^*=A_i\subseteq A^+_{ H}\subseteq A_H^*$. If $G$ is a multi-agent group, then suppose that $d\in A_G^*$. If $d\in A^+_{ G}$, then since $ A^+_{ G}\subseteq A^+_{ H}$, we have $d\in A^+_{ H}\subseteq A_H^*$. If $d=\la d_0,\cdots, d_n\ra\in A_{G_0}^*\times\cdots\times A_{G_n}^*$ where $\{G_0,\cdots,G_n\}$ is a non-trivial partial partition of $G$ and $G_0\prec\cdots\prec G_n$, since $G\subseteq H$, $\{G_0,\cdots,G_n\}$ is also a non-trivial partial partition of $H$. Then $\la d_0,\cdots, d_n\ra\in A_H^*$. Therefore $A_G^*\subseteq A_H^*$.
\end{proof}

By Proposition \ref{propMonoAct}, $A_I^*$ is the set of all distributed actions and $A_I^+\subseteq A_I^*$. We now define distributed transitions according to the observation O3 in the introduction.
\begin{definition}[Distributed transition]
    For each $d=\la d_0,\cdots, d_n\ra\in A_I^*$, we define the distributed transition relation $\rel{d}:=\bigcap_{0\leq k\leq n}\rel{d_k}$.
\end{definition}
For any nonempty $X,Y\subseteq S$ and $d\in A_I^*$, we use $X\rel{d}Y$ to indicate that there is some $s\in X$ and some $t\in Y$ such that $s\rel{d}t$. A distributed transition relation could be empty, which indicates that there are conflicting actions among subgroups. For example, the actions of opening and closing the same door cannot constitute an executable distributed action.

Although a joint action $d=\la d_0,\cdots, d_n\ra\in A_I^*$ may well use other joint actions as $d_i$, they can be in effect reduced to the intersections of atomic transitions, e.g., the transitions for $\lr{\lr{a,b},c}$ are the same as the transitions for $\lr{a,b,c}$, as they are computed by taking the intersection of $\rel{a}$, $\rel{b}$, and $\rel{c}$. The following observation plays an important role in the completeness proof.
\begin{proposition}\label{propAlterDef}
    For any nonempty group $G$ and $d\in A_G^*$, there exist $a_0'\in A_{G_0'},\cdots,a_m'\in A_{G_m'}$ such that $\{G_0',\cdots,G_m'\}$ is a partial partition of $G$ and $\rel{d}=\bigcap_{0\leq k\leq m} \rel{a_k'}$.
\end{proposition}

\begin{proof}
    We show by induction on the number of agents in $G$. If $G=\{i\}$ where $i\in I$, then $d\in A_i$ and $\{\{i\}\}$ is a partial partition of $\{i\}$. Assume that $G$ is a multi-agent group. If $d\in A^+_{ G}$, then there exists $G'\subseteq G$ such that $d\in A_{G'}$,  then $\{G'\}$ is a partial partition of $G$. If $d=\la d_0,\cdots, d_n\ra\in A_{G_0}^*\times\cdots\times A_{G_n}^*$ where $\{G_0,\cdots,G_n\}$ is a non-trivial partial partition of $G$ and $G_0\prec\cdots\prec G_n$, then $n>0$ and $\emptyset\ne G_k\subsetneq G$ for $0\leq k\leq n$. By induction hypothesis, for $0\leq k\leq n$ there exist $a_{0}^k\in A_{G_{0}^k},\cdots,a_{m_k}^k\in A_{G_{m_k}^k}$ such that $\{G_{0}^k,\cdots,G_{m_k}^k\}$ is a partial partition of  $G_k$ and $\rel{d_k}=\bigcap_{0\leq j\leq m_k} \rel{a_j^k}$. Then $\{G_{0}^0,\cdots,G_{m_0}^0,\cdots,G_{0}^n,\cdots,G_{m_n}^n\}$ is a partial partition of $G$. Since $\rel{\la d_0,\cdots, d_n\ra}=\bigcap_{0\leq k\leq n}\rel{d_k}$, the proposition holds.
\end{proof}




Observe that in the last case of the proof, $\lvert\{G_{0}^0,\cdots,G_{m_0}^0,\cdots,G_{0}^n,\cdots,G_{m_n}^n\}\rvert>1$ because $n>0$. Therefore, we have the following corollary, which will be helpful for the completeness proof:

\begin{corollary}\label{CorAlterDef}
    For any multi-agent group $G$ and $d\in A_G^*\setminus A_G^+$, there exist $a_0\in A_{G_0},\cdots,a_m\in A_{G_m}$ such that $m>0$, $\{G_0,\cdots,G_m\}$ is a partial partition of $G$ and $\rel{d}=\bigcap_{0\leq k\leq m} \rel{a_k}$.
\end{corollary}

The following definitions are generalizations of the counterparts in \cite{fervari2017strategically}.

\begin{definition}[Executability]
    For any nonempty $X\subseteq S$, we say that $d\in A_I^*$ is executable on $X$, if for each $s\in X$ there exists $t$ such that $s\rel{d}t$.
\end{definition}
\begin{definition}[Strategy]
  A strategy of group $G$ is a partial function $\sigma_G:[S]_G \to A_G^*$ such that $\sigma_G([s]_G)$ is executable on $[s]_G$. Particularly, for each group, the empty function is also a strategy, called the empty strategy.
\end{definition}
By definition, the only strategy for the empty group is the empty strategy.
\begin{definition}[Execution]\label{Execution}
    Given a strategy $\sigma_G$ of group $G$ w.r.t. a model $\M$, a \emph{possible execution} of $\sigma_G$ is a possibly infinite nonempty sequence of equivalence classes $\delta=[s_0]_G[s_1]_G\cdots$ such that $[s_j]_G\xrightarrow{\sigma_G([s_j]_G)}[s_{j+1}]_G$ for all $0\leq j<\lvert\delta\rvert -1$. If the execution is a finite sequence $[s_0]_G\cdots [s_n]_G$, we call $[s_n]_G$ the leaf-node, and $[s_j]_G(0\leq j<n)$ an inner-node w.r.t. this execution. If it is infinite, then all $[s_j]_G(j\in\mathbb{N})$ are inner-nodes. A possible execution of $\sigma_G$ is \emph{complete} if it is infinite or its leaf-node is not in $\Dom(\sigma_G)$. We use $\CELeafN(\sigma_G, s)$ to denote the set of all leaf-nodes of all complete executions of $\sigma_G$ starting from $[s]_G$, and $\CEInnerN(\sigma_G, s)$ to denote the set of all inner-nodes of all complete executions of $\sigma_G$ starting from $[s]_G$.
  \end{definition}
Now we formally define the semantics of know-how $\Kh_G\phi$ which intuitively says that there is a strategy $\sigma_G$ such that $G$ distributedly knows that $\sigma$ will terminate and guarantee $\phi$.
\begin{definition}[Semantics]\label{Semantics}
  Given a model $\M$, for any state $s\in S$ and any formula $\varphi\in\DKH$
  \begin{center}
    \begin{tabular}{|lcl|}
    \hline
    $\M,s \vDash \top$ & & always holds \\
    $\M,s \vDash p$ & iff & $s \in V(p)$, where $p \in P$ \\
    $\M,s \vDash \neg \varphi$ & iff & $\M,s \nvDash \varphi$ \\
    $\M,s \vDash \varphi \land \psi$ & iff & $\M,s \vDash \varphi$  and  $\M,s \vDash \psi$ \\
    $\M,s \vDash \K_G \varphi$ & iff & $\M,s' \vDash \varphi$ for all $s' \in [s]_G$ \\
    $\M,s \vDash \Kh_G \varphi$ & iff & 
    \begin{tabular}[t]{@{}l@{}}
        there is a strategy $\sigma_G$ of $G$ such that: \\
        1. $[t]_G \subseteq \llrr{\varphi}$ for all $[t]_G \in \CELeafN(\sigma_G, s)$, and \\
        2. all its complete executions starting from $[s]_G$ are finite.
    \end{tabular}\\
        \hline
\end{tabular}
\end{center}
    where $\llrr{\phi}=\{s\in S\mid \M,s\vDash\phi\}$.
\end{definition}

        
        
            
        
Note that when $G$ is a singleton, the semantics of $\Kh_G$ is exactly as in \cite{fervari2017strategically}. It is also worth noting that the key axiom Cooperation:\footnote{An alternative form is $\Kh_G(\phi\to\psi)\to(\Kh_H\phi\to\Kh_{G\cup H}\psi)$.} $\Kh_G\phi\land\Kh_H\phi\to\Kh_{G\cup H}(\phi\land \psi)$  ($G$ and $H$ are disjoint) in the coalition-based approach such as \cite{naumov2018together} is \textit{not} valid in our setting, which demonstrates that our mixed framework is by no means a trivial combination. There are two reasons for the invalidity. First of all, as our know-how is based on \textit{multi-step} strategies, the strategies behind the know-how of $G$ and $H$ may not be synchronizable in the number of steps of reaching the goals, which is fundamentally different from the one-step coalition-based setting where joint actions always end at the same time. Moreover, even if we restrict to one-step actions, the distributed transition relations (as intersections) can be empty, that is, joint actions may be non-executable, which is different from the setting in \cite{naumov2018together} and makes it uncertain whether the two actions can be combined into an executable joint action. 


Moreover, the invalidity also shows that distributed knowledge-how cannot be conceptualized as simply putting together individual/subgroup knowledge-how (and their consequences), which separates it from distributed knowledge-that, where Cooperation is valid.

As a technical remark, the empty group provides us with additional expressive power of a \textit{universal modality}, which will play a vital role in the axiomatization proposed in the next section.
\begin{proposition}\label{Khempty}
    $\M,s\vDash\K_\emptyset\phi$ iff $\M,s\vDash\Kh_\emptyset\phi$ iff $\M,t\vDash\phi$ for any $t\in S$.
\end{proposition}

\begin{proof}
It is obvious that $\M,s\vDash\K_\emptyset\phi$ iff $\M,t\vDash\phi$ for any $t\in S$. We show $\M,s\vDash\Kh_\emptyset\phi$ iff $\M,t\vDash\phi$ for any $t\in S$.

    $(\Rightarrow)$ Since $\M,s\vDash\Kh_\emptyset\phi$, then there is a strategy $\sigma_\emptyset$ for the empty group such that
        \begin{enumerate}
        \item $[t]_\emptyset\subseteq\llrr{\varphi}$ for all $[t]_\emptyset\in\CELeafN(\sigma_\emptyset,s)$, and 
        \item all its complete executions starting from $[s]_\emptyset$ are finite.
        \end{enumerate}
    By definition, the only strategy for the empty group is the empty strategy $\emptyset$. Since $[s]_\emptyset\notin\Dom(\emptyset)$, the only execution of the empty strategy $\emptyset$ starting from $[s]_\emptyset$ is $\delta_\emptyset=[s]_\emptyset$, then $\CELeafN(\sigma_\emptyset,s)=\{[s]_\emptyset\}$. Therefore, $S=[s]_\emptyset\subseteq\llrr{\phi}$.

    $(\Leftarrow)$ Since $\M,t\vDash\phi$ for any $t\in S$, we have $\M,s\vDash\K_\emptyset\phi$. We show a stronger result that $\M,s\vDash\K_G\phi$ implies $\M,s\vDash\Kh_G\phi$ for any group $G$. Consider the empty strategy $\emptyset$ for $G$. Since $[s]_G\notin\Dom(\emptyset)$, the only execution of the empty strategy $\emptyset$ starting from $[s]_G$ is $\delta_G=[s]_G$, then $\CELeafN(\emptyset,s)=\{[s]_G\}$. Since $\M,s\vDash\K_G\phi$, we have $[s]_G\subseteq\llrr{\phi}$. Therefore, $\M,s\vDash\Kh_G\phi$.
\end{proof}



\section{Axiomatization}\label{sec.axiom}
\begin{center}
\textbf{Proof system $\SDKh$}\footnote{Formulas $\phi,\psi$ are from the language defined in Definition \ref{defLang}.}

\begin{longtable}{llll}
    
    \textbf{Axioms} & &  & \\
    TAUT & axioms of propositional logic\footnote{All instantiations in language $\DKH$ of axioms of propositional logic, to be more precise.} & AxKtoKh & $\K_G\phi\to\Kh_G\phi$\\
    DISTK & $\K_G \phi\land\K_G(\phi\to \psi)\to\K_G \psi$ &AxEmpKhtoK & $\Kh_\emptyset\phi\to\K_\emptyset\phi$\\
    T & $\K_G \phi\to \phi$ & AxKhtoKKh & $\Kh_G\phi\to\K_G\Kh_G\phi$\\
    4 & $\K_G \phi\to\K_G\K_G \phi$ & AxEmpMono & $\K_\emptyset(\phi\to\psi)\to\K_\emptyset(\Kh_G\phi\to\Kh_G\psi)$\\
    5 & $\neg\K_G \phi\to\K_G\neg\K_G \phi$ & AxKhbot & $\Kh_G\bot\to\bot$\\
    AxKMono & $\K_G \phi\to\K_H \phi$, where $G\subseteq H$ & AxKhtoKhK & $\Kh_G\phi\to\Kh_G\K_G\phi$\\
    AxKhMono & $\Kh_G \phi\to\Kh_H \phi$, where $G\subseteq H$ & AxKhKh & $\Kh_G\Kh_G\phi\to\Kh_G\phi$\\

    \textbf{Rules} & & &\\
    \multirow{2}{*}{MP} & \multirow{2}{*}{$\dfrac{\phi,\phi\to\psi}{\psi}$} & \multirow{2}{*}{NECK} & \multirow{2}{*}{$\dfrac{\vdash\phi}{\vdash\K_G\phi}$}\\
    &&&
      \end{longtable}
    \end{center}

$\SDKh$ contains all $\mathbb{S}5$ axioms for $\K$. AxKMono and AxKhMono are monotonicity axioms for distributed knowledge-that and distributed knowledge-how. They express that a group knows/knows how $\phi$ whenever its subgroup knows/knows how $\phi$. Their validity is due to the monotonicity of distributed indistinguishability and distributed action sets shown by Proposition~\ref{propMonoAct}. We will provide the validity of AxKhMono later. AxKhKh is a key axiom w.r.t.\ multi-step planning, revealing the compositional nature of knowing how. Its validity is highly non-trivial as in \cite{fervari2017strategically}, and we will provide a sketch of the proof later.

Note that the logic is \textit{not} normal, for example, $\Kh_G\phi\land\Kh_G\psi\to\Kh_G(\phi\land\psi)$ is invalid. In \cite{fervari2017strategically} the monotonicity rule MONOKh is proposed to fill the gap left by the absence of the K axiom for $\Kh$, which states that if $\phi\to\psi$ is provable then $\Kh\phi\to\Kh\psi$ is provable. In $\SDKh$ we use a stronger \textit{axiom} AxEmpMono which reflects the monotonicity of $\Kh_G$ at the \textit{model level}. The rule MONOKh can be derived from AxEmpMono, NECK and T.

As for other important axioms, AxKtoKh says if a group already knows $\phi$ then it knows how to achieve $\phi$ by doing nothing (the empty strategy). AxEmpKhtoK helps us reduce $\Kh_\emptyset$ to $\K_\emptyset$. AxKhtoKKh is the positive introspection axiom for $\Kh_G$, whose validity comes from uniformity of strategies, and the negative introspection of know-how is derivable. AxKhbot encodes part of the termination condition within the semantics, for it excludes strategies that have no terminating executions. AxKhtoKhK says that $\phi$ is known after executing the 
know-how strategy, as required in the semantics, which is usually also assumed in discussions of contingent planning (cf. \cite{li2017more}).


\begin{proposition}[Axiom AxKhMono]
    $\M,s\vDash\Kh_G\phi$ implies $\M,s\vDash\Kh_H\phi$, where $G\subseteq H$.
\end{proposition}

\begin{proof}
    Suppose $\M, s \vDash \Kh_G \phi$. Then there exists a strategy $\sigma_G$ such that:
\begin{enumerate}
    \item $[t]_G \subseteq \llrr{\varphi}$ for all $[t]_G \in \CELeafN(\sigma_G, s)$, and
    \item all its complete executions starting from $[s]_G$ are finite.
\end{enumerate}

If $[s]_G \notin \Dom(\sigma_G)$, the only execution of $\sigma_G$ starting from $[s]_G$ is $\delta_G = [s]_G$, which is also complete and finite. Then $\CELeafN(\sigma_G, s) = \{[s]_G\}$, implying $[s]_G \subseteq \llrr{\varphi}$. Since $G\subseteq H$, it follows that $[s]_H\subseteq [s]_G\subseteq\llrr{\varphi}$, then by empty strategy we have $\M, s \vDash \Kh_H \phi$.

If $[s]_G \in \Dom(\sigma_G)$, then $A_G^*$ is non-empty. Since $G \subseteq H$, by Proposition \ref{propMonoAct} we have $A_G^* \subseteq A_H^*$, so $A_H^*$ is also non-empty. Define partial function $\sigma_H: [S]_H \to A_H^*$ as follows: for each $v \in S$ such that $[v]_G \in \Dom(\sigma_G)$, let $\sigma_H([v]_H) = \sigma_G([v]_G)$. We first show that $\sigma_H$ is well-defined. Suppose $s_1 \sim_H s_2$. Since $G \subseteq H$, we have $s_1 \sim_G s_2$, so $[s_1]_G = [s_2]_G$. Thus $[s_1]_G \in \Dom(\sigma_G)$ if and only if $[s_2]_G \in \Dom(\sigma_G)$. If $[s_1]_G = [s_2]_G \in \Dom(\sigma_G)$, then $\sigma_H([s_1]_H) = \sigma_G([s_1]_G) = \sigma_G([s_2]_G) = \sigma_H([s_2]_H)$. Additionally, for any $[v]_H \in \Dom(\sigma_H)$, by definition $[v]_G \in \Dom(\sigma_G)$ and $\sigma_G([v]_G)$ is executable on $[v]_G$. Since $\sigma_H([v]_H) = \sigma_G([v]_G)$ and $[v]_H \subseteq [v]_G$, $\sigma_H([v]_H)$ is executable on $[v]_H$. Therefore, $\sigma_H$ is well-defined.

Let $\delta_H = [s_0]_H [s_1]_H \dots$ be an arbitrary complete execution of $\sigma_H$ starting from $[s]_H = [s_0]_H$, with $\lvert \delta_H \rvert > 1$. We show that $\delta_H$ is finite and its leaf nodes are contained in $\llrr{\phi}$. For $[s_j]_H$ in $\delta_H$ such that $[s_j]_H\in\Dom(\sigma_H)$, we have $[s_j]_H \xrightarrow{\sigma_H([s_j]_H)} [s_{j+1}]_H$. Thus there exist $s_j' \in [s_j]_H$ and $s_{j+1}' \in [s_{j+1}]_H$ such that $s_j' \xrightarrow{\sigma_H([s_j]_H)} s_{j+1}'$. By the definition of $\sigma_H$, $\sigma_H([s_j]_H) = \sigma_G([s_j]_G)$, so $s_j' \xrightarrow{\sigma_G([s_j]_G)} s_{j+1}'$. This implies $[s_j']_G \xrightarrow{\sigma_G([s_j]_G)} [s_{j+1}']_G$. Since $s_j \sim_H s_j'$, $s_{j+1} \sim_H s_{j+1}'$, and $G \subseteq H$, we have $s_j \sim_G s_j'$, $s_{j+1} \sim_G s_{j+1}'$, hence $[s_j]_G = [s_j']_G$, $[s_{j+1}]_G = [s_{j+1}']_G$. Therefore, $[s_j]_G \xrightarrow{\sigma_G([s_j]_G)} [s_{j+1}]_G$. Let $\delta_G = [s_0]_G [s_1]_G \dots$. Then $\delta_H$ must be finite; otherwise, $\delta_G$ would be an infinite execution of $\sigma_G$ starting from $[s]_G$. Let $\lvert \delta_G \rvert = \lvert \delta_H \rvert = n + 1$. Since $\delta_H$ is finite and complete, $[s_n]_H \notin \Dom(\sigma_H)$, implying $[s_n]_G \notin \Dom(\sigma_G)$. Thus $\delta_G$ is a finite complete execution starting from $[s]_G$. As $\M, s \vDash \Kh_G \phi$, we have $[s_n]_G \subseteq \llrr{\phi}$, hence $[s_n]_H \subseteq \llrr{\phi}$. Therefore, $\M, s \vDash \Kh_H \phi$.
\end{proof}

\begin{proposition}[Axiom AxKhKh]\label{khkh}
     $\M,s\vDash\Kh_G\Kh_G\phi$ implies $\M,s\vDash\Kh_G\phi$.
\end{proposition}

\begin{proof}
    [Sketch] Axiom AxKhKh states that for any pointed model $\M,s$, if there exists a strategy $\sigma_G$ of $G$ such that all complete executions starting from $[s]_G$ are finite, and for each leaf node $[s']_G$ of these executions, there exists another strategy $\sigma_G'$ of $G$ such that all complete executions starting from $[s']_G$ are finite and $\phi$ is satisfied on their leaf nodes, then there exists a strategy $\tau_G$ of $G$ such that all complete executions starting from $[s]_G$ are finite and $\phi$ is satisfied on their leaf nodes. To prove the validity, we need to combine $\sigma_G$ and each $\sigma_G'$ into $\tau_G$, while resolving three critical issues: 

First, the domains of strategies $\sigma_G'$ may contain irrelevant equivalence classes that are not visited during executions starting from its corresponding leaf node; however, when combined with other strategies, new executions might pass through these equivalence classes. This may result in the new execution failing to achieve $\phi$. Second, overlapping domains of different strategies could result in the combined strategy not being a function. Third, the domain of some strategies might include leaf nodes of other strategies' complete executions. This can cause executions that are originally planned to terminate to continue, even leading to infinite executions.

To address these, we will proceed as follows: 
\begin{enumerate}
\item Restrict the domain of each strategy to the inner nodes of its complete executions starting from the corresponding leaf node, eliminating irrelevant equivalence classes. 
\item Use the Well-ordering Theorem to order the leaf nodes of complete executions of $\sigma_G$ (this also orders their corresponding strategies $\sigma_G'$), then combine these strategies according to the ordering while skipping defined equivalence classes and leaf nodes, constructing $\tau_G^\gamma$. 
\item Finally, combine $\sigma_G$ and $\tau_G^\gamma$ into $\tau_G$, also skipping defined equivalence classes and leaf nodes.
\end{enumerate}

We present a sketch of the proof (cf. \cite{fervari2017strategically} for a similar proof). Suppose that $\M,s\vDash\Kh_G\Kh_G\phi$, then there exists a strategy $\sigma_G$ of $G$ such that all complete executions starting from $[s]_G$ are finite and $\Kh_G\phi$ is satisfied on their leaf nodes. By the Well-ordering Theorem, let $\CELeafN(\sigma_G, s)=\left\{S_i \mid i<\gamma\right\}$, where $\gamma$ is an ordinal number. Pick $s_i\in S_i$, then $\left[s_i\right]_G=S_i$. Since $\M,s_i \vDash \Kh_G \varphi$ for each $i<\gamma$, there exists a strategy $\sigma_G^i$ of $G$ for each $i<\gamma$ such that all complete executions starting from $[s_i]_G$ are finite and $\phi$ is satisfied on their leaf nodes. We define $\tau_G^\gamma$ as follows:

Define $\tau_G^i$ inductively where $0 \leq i<\gamma$:
\begin{itemize}
\item If $i=0$, $\tau_G^0=\sigma_G^0|_{\text{CEInner}(\sigma_G^0, s_0)}$;
\item If $i>0$, $\tau_G^i=f_i \cup(\sigma_G^i|_{D_i})$,
\end{itemize}
where $f_i=\bigcup_{j<i} \tau_G^j$, $D_i=$ $\CEInnerN\left(\sigma_G^i, s_i\right) \setminus(\Dom\left(f_i\right)\cup \{[v]_G \in \CELeafN\left(f_i, t\right) \mid[t]_G \in  \Dom\left(f_i\right)\})$. Let $\tau_G^\gamma=\bigcup_{i<\gamma} \tau_G^i$.

We can see that the domain of each $\tau_G^i$ is restricted to the inner nodes of its complete executions starting from the corresponding leaf node, and defined equivalence classes and leaf nodes are skipped (this requires further proof). Moreover, we can verify that $\tau_G^\gamma$ is indeed a partial function. With these results, we can show that on any equivalence class of $\Dom(\tau_G^\gamma)$, $\tau_G^\gamma$ is a witness of $\Kh_G\phi$.

Finally, we define strategy $\tau_G$:
$$\tau_G=\tau_G^\gamma \cup(\sigma_G|_E),$$
where $E= \CEInnerN(\sigma_G, s) \setminus(\Dom(\tau_G^\gamma) \cup\{[v]_G \in \CELeafN(\tau_G^\gamma, t) \mid[t]_G \in \Dom(\tau_G^\gamma)\})$. The combination of $\tau_G^\gamma$ and $\sigma_G$ follows a similar method as before, also skipping defined equivalence classes and leaf nodes. With the above results, we can show that $\tau_G$ is the desired strategy.

\end{proof}

The above propositions give us the soundness of the logic.  
\begin{theorem}[Soundness]
    $\vdash\phi$ implies $\vDash\phi$.
\end{theorem}

\section{Completeness} \label{sec.comp}


In this section, we prove the completeness theorem for $\SDKh$. Inspired by \cite{naumov2018together}, we adopt the unraveling technique to handle distributed knowledge-that. On the other hand, when dealing with distributed knowledge-how, we face a similar problem to that with distributed knowledge-that. To find witness strategies for $\Kh_{G}$-formulas, we make use of atomic group actions and define them as pairs $(\phi, G)$, and then transition relations can be defined as $s\rel{(\phi,G)}t$ iff $\Kh_G\phi\in s$ and $\K_G\phi\in t$, meaning that the witness for $\Kh_G\phi$ can be found in one step by $(\phi, G)$. However, $\Kh_{\{i,j\}}(\phi\land\psi)$ could be satisfied at 
a state $s$ accidentally by $\lr{(\phi,i),(\psi,j)}$ in the absence of $\Kh_{\{i,j\}}(\phi\land\psi)\in s$, since $\K_{\{i,j\}}(\phi\land\psi)\in t$ always holds whenever $s\rel{\lr{(\phi,i),(\psi,j)}}t$. To avoid this, we use the unraveling technique again, splitting distributed transitions into transitions of atomic group actions, thereby ensuring that all joint actions are not executable in the canonical model. Moreover, our multi-step setting introduces complications when refuting $\Kh_{G}$-formulas in the canonical model, where AxKhKh and AxEmpMono will play crucial roles instead of the Cooperation axiom in \cite{naumov2018together}.

Based on the ideas discussed above, we fix a maximal consistent set $X_0$ and define the canonical model $\M^c(X_0)$ as follows:
\begin{definition}[Canonical model]
    Given a maximal consistent subset $X_0$ of $\DKH$,  the canonical model $\M^c(X_0)=\left\la S^c, \left\{\sim_i^c\right\}_{i\in  I},\{A_G^c\}_{G\subseteq I}, \{\rel{a}\mid a\in A_I^{c+}\}, V^c\right\ra$ is defined as follows:
    \begin{itemize}
        \item $S^c$ is the set of all \emph{mixed sequences} $X_0\la(\phi_1,G_1),H_1\ra X_1\cdots \la(\phi_n,G_n),H_n\ra X_n$ such that 
        \begin{itemize}
            \item $X_j$ is a maximal consistent subsets of $\DKH$ for each $j\geq 0$,
            \item $\phi_j\in\DKH$ for each $j\geq 1$,
            \item $G_j\ne\emptyset$ and $G_j,H_j\subseteq I$ for each $j\geq 1$,
            \item $\Kh_{G_j}\phi_j\in X_{j-1}$ and $\K_{G_j}\phi_j\in X_j$ for each $j\geq 1$,
            \item $\{\phi\mid\K_{H_j}\phi\in X_{j-1}\}\subseteq X_j$ for each $j\geq 1$;
        \end{itemize}
        
        for any state $s=X_0\la(\phi_1,G_1),H_1\ra X_1\cdots \la(\phi_n,G_n),H_n\ra X_n$, we use $ed(s)$ to denote $X_n$;
        \item For each $i\in I$, $\sim_i^c$ is defined as follows: for any states $s=X_0\la(\phi_1,G_1),H_1\ra X_1\cdots \la(\phi_n,G_n),H_n\ra X_n$ and $s'=X_0\la(\phi_1',G_1'),H_1'\ra X_1'\cdots \la(\phi_m',G_m'),H_m'\ra X_m'$, $s\sim_i^c s'$ iff there exists an integer $k$ such that:
        \begin{itemize}
            \item $0\leq k\leq \min\{n,m\}$;
            \item $X_j=X_j^{\prime}$ and $H_j=H_j^{\prime}$ for each $1 \leq j \leq k$;
            \item $i\in H_j$ for each $k<j \leq n$;
            \item $i\in H_j'$ for each $k<j \leq m$;
        \end{itemize}
        \item For each nonempty $G\subseteq I$, $A_G^c=\{(\phi, G)\mid \phi\in\DKH\}$ and $A_\emptyset^c=\emptyset$;
        \item For any $(\phi,G)\in A_I^{c+}$ and states $s,t$, $s\rel{(\phi,G)}t$ iff $t=s\la(\phi,G),G'\ra X$ where $G'\subseteq I$ and $X$ is an MCS;
        \item For any $p\in P$, $V^c(p)=\{s\in S^c\mid p\in ed(s)\}$.
    \end{itemize}
\end{definition}

Before explaining the construction of the canonical model, we first verify that $\M^c(X_0)$ is well-defined, i.e., $\sim_i^c$ is an equivalence relation and $A^c_G$ satisfies the requirement in Definition \ref{defmodel}. 
\begin{proposition}\label{propEqRel}
    $\sim^c_i$ is an equivalence relation on $S^c$ for any $i\in I$.
\end{proposition}


\begin{proof}

Let $ i \in I $. $\sim^c_i$ clearly satisfies reflexivity and symmetry. 

For any states $s=X_0\la(\phi_1,G_1),H_1\ra X_1\cdots \la(\phi_n,G_n),H_n\ra X_n$,
$s'=X_0\la(\phi_1',G_1'),H_1'\ra X_1'\cdots \la(\phi_m',G_m'),H_m'\ra X_m'$, and
$ s'' = X_0\la(\phi_1'',G_1''),H_1''\ra X_1''\cdots \la(\phi_r'',G_r''),H_r''\ra X_r''$,
suppose $ s \sim_i^c s' $ and $ s' \sim_i^c s'' $. Then there exist integers $ k $ and $ k' $ such that:

\begin{itemize}
    \item $ 0 \leq k \leq \min \{n, m\} $, satisfying:
    \begin{itemize}
        \item $ X_j = X_j' $ and $ H_j = H_j' $ for $ 1 \leq j \leq k $;
        \item $ i \in H_j $ for $ k < j \leq n $;
        \item $ i \in H_j' $ for $ k < j \leq m $;
    \end{itemize}
    \item $ 0 \leq k' \leq \min \{m, r\} $, satisfying:
    \begin{itemize}
        \item $ X_j' = X_j'' $ and $ H_j' = H_j'' $ for $ 1 \leq j \leq k' $;
        \item $ i \in H_j' $ for $ k' < j \leq m $;
        \item $ i \in H_j'' $ for $ k' < j \leq r $.
    \end{itemize}
\end{itemize}

Let $ k'' = \min \{k, k'\} $. Then for all $ j $:
\begin{itemize}
    \item If $ 1 \leq j \leq k'' $, then $ X_j = X_j'' $ and $ H_j = H_j'' $;
    \item If $ k'' < j \leq n $:
    \begin{itemize}
        \item Case $ k \leq k' $: $ k'' = k $, and then $ i \in H_j $ for $ k'' < j \leq n $;
        \item Case $ k' < k $: $ k'' = k' $, and then:
        \begin{itemize}
            \item For $ k < j \leq n $, $ i \in H_j $;
            \item For $ k' < j \leq k $, since $k\leq m$, we have $ H_j = H_j' $ and $ i \in H_j' $, and then $ i \in H_j $.
        \end{itemize}
    \end{itemize}
    \item If $ k'' < j \leq r $, similarly $ i \in H_j'' $.
\end{itemize}

Therefore, $ s \sim_i^c s'' $. This proves that $\sim^c_i$ is an equivalence relation on $ S^c $.
\end{proof}
By the definition of $A^c_G$ it is obvious: 
\begin{proposition}
    $G\subsetneq H$ implies $A_G^c\cap A_H^c=\emptyset$.
\end{proposition}

Now we are ready to explain how we construct the canonical model. Recall that we would like to unravel both the (group) epistemic relations and the (group) action relations in the model. Instead of building a model and then unraveling it, we do it in one go by considering paths of MCSs. Note that between two given MCSs, there can be both the epistemic relations and the action relations, and this is the reason for considering the so-called  mixed sequences $X_0\la(\phi_1,G_1),H_1\ra X_1\cdots \la(\phi_n,G_n),H_n\ra X_n$, where $H_i$ represents the epistemic relation and $(\phi_i,G_i)$ represents the action transition. Moreover, although we need to unravel the epistemic relations and the action transitions respectively, we do not need to unravel the ``cross-type'' relations, i.e., it is fine to have transitions of both types between two MCSs. The detailed definition of the mixed sequences resembles the usual conditions of canonical epistemic relation and the canonical relations for the actions in the literature of knowing how logics \cite{fervari2017strategically}.

Interested readers may wonder why in mixed sequences, we can always require the double (epistemic+action) transitions between the MCSs. Intuitively, there can be two MCSs with only epistemic relations between them, without any action transition (or the other way around). Thanks to the epistemic relation of the empty group and the trivial action $(\top, G)$, such cases are not ruled out, as there are always such trivial epistemic and action relations between any two relevant MCSs: the empty group cannot distinguish any two states and any group always knows how to achieve $\top$. The trivial epistemic relation $\sim_\emptyset$ and the trivial action transition $\rel{(\top, G)}$ will also play an important role in the existence lemmas for both $\K_G$ and $\Kh_G$.

Finally, the definition of $\sim^c_G$ is essentially the same as the canonical epistemic relation in \cite{naumov2018together}, if we ignore the action transitions in the mixed sequences. Intuitively, two sequences are indistinguishable by $i$ if they are two extensions of the same initial history (ignoring the action information) such that $i$ is always in the groups that cannot distinguish the adjacent MCSs from the departing point onward of the two extensions. The action transitions between the mixed sequences are defined similarly as the canonical action relation used in \cite{fervari2017strategically}. 
\medskip

We now show that all joint actions are not executable in the canonical model, as discussed in the beginning of this section.


\begin{proposition}\label{PropIntNonexec}
    For any $d\in A_I^{c*}\setminus A_I^{c+}$, $d$ is not executable anywhere in $\M^c$.
\end{proposition}

\begin{proof}
    By Corollary \ref{CorAlterDef}, there exist $(\phi_0,G_0)\in A_{G_0}^c,\cdots,(\phi_m,G_m)\in A_{G_m}^c$ such that $m>0$, $\{G_0,\cdots,G_m\}$ is a partial partition of $I$ and $\rel{d}=\bigcap_{0\leq k\leq m} \rel{(\phi_k,G_k)}$, then $G_0,\cdots,G_m$ are distinct. Suppose towards a contradiction that there is $s,t\in S^c$ such that $s\rel{d}t$. 
    Then $s\rel{(\phi_k,G_k)}t$ for $0\leq k\leq m$. Then $t=s\la(\phi_0,G_0),G_0'\ra Y_0=\cdots=s\la(\phi_m,G_m),G_m'\ra Y_m$ where $G_0',\cdots,G_m'\subseteq I$ and $Y_0,\cdots,Y_m$ are MCSs. But $G_0,\cdots,G_m$ are distinct, contradiction.
\end{proof}

For the knowledge-that part in the Truth Lemma, we firstly need the following propositions:

\begin{proposition}\label{Kntok}
For any $s=X_0\la(\phi_1,G_1),H_1\ra X_1\cdots \la(\phi_n,G_n),H_n\ra X_n \in S^c$ and $k \leq n$, if $\K_H\phi \in X_n$ and $H \subseteq H_j$ holds for all $k < j \leq n$, then $\K_H\phi \in X_k$.
\end{proposition}

\begin{proof}
Suppose there exists $k \leq n$ such that $\neg\K_H\phi \in X_k$. Let $m$ be the maximal such $k$. Since $\K_H\phi \in X_n$, we have $m < n$. By Axiom 5, we get $\K_H\neg\K_H\phi \in X_m$. By Axiom AxKMono, we have $\K_{H_{m+1}}\neg\K_H\phi \in X_m$. By the definition of states in the canonical model, it follows that $\neg\K_H\phi \in X_{m+1}$, contradicting the maximality of $m$.
\end{proof}


\begin{proposition}\label{Kkton}
For any $s = X_0\la(\phi_1,G_1),H_1\ra X_1\cdots \la(\phi_n,G_n),H_n\ra X_n \in S^c$ and $k \leq n$, if $\K_H\phi \in X_k$ and $H \subseteq H_j$ holds for all $k < j \leq n$, then $\phi \in X_n$.
\end{proposition}

\begin{proof}
We show by induction on $n - k$. If $n - k = 0$, then $\K_H\phi \in X_n$. By Axiom T, we have $\phi \in X_n$. If $n - k > 0$, by Axiom 4 and Axiom AxKMono, we get $\K_{H_{k+1}}\K_H\phi \in X_k$. By the definition of states in the canonical model, it follows that $\K_H\phi \in X_{k+1}$. Since $n - (k+1) < n - k$, $k+1 \leq n$, $\K_H\phi \in X_{k+1}$, and $H \subseteq H_j$ holds for all $k+1 < j \leq n$, by induction hypothesis we have $\phi \in X_n$.
\end{proof}


With Proposition \ref{Kntok} and \ref{Kkton}, we can show the following key lemmas that are essential for the knowledge-that part in the Truth Lemma. Note that we will need to make use of the trivial action transition $\rel{(\top,G)}$ in the proof of the existence lemma, as mentioned earlier in this section.

\begin{lemma}\label{KShare}
If $\K_H \varphi \in ed(s)$ and $s \sim_H^c s^{\prime}$, then $\varphi \in ed\left(s^{\prime}\right)$.
\end{lemma}

\begin{proof}
Note that $\sim_H^c$ is defined by Definition \ref{group indistinguishability} via indistinguishability relations of agents. Let $s = X_0\la(\phi_1,G_1),H_1\ra X_1\cdots \la(\phi_n,G_n),H_n\ra X_n$ and $s' = X_0\la(\phi_1',G_1'),H_1'\ra X_1'\cdots \la(\phi_m',G_m'),H_m'\ra X_m'$. If $H=\emptyset$, then since $\emptyset\subseteq H_j$ holds for all $0<j\leq n$, by Proposition \ref{Kntok} we have $\K_\emptyset\phi\in X_0$. Since $\emptyset\subseteq H_j'$ holds for all $0<j\leq m$, by Proposition \ref{Kkton} we have $\phi\in X_m'=ed(s')$.

If $H\ne\emptyset$, since $s \sim_H^c s^{\prime}$, for each $i\in H$ there exists an integer $k_i$ satisfying:
\begin{itemize}
    \item $0 \leq k_i \leq \min \{n, m\}$;
    \item $X_j = X_j^{\prime}$ and $H_j = H_j^{\prime}$ for $1 \leq j \leq k_i$;
    \item $i \in H_j$ for $k_i < j \leq n$;
    \item $i \in H_j'$ for $k_i < j \leq m$.
\end{itemize}
Let $k=\max\{k_i\}_{i\in H}$. Then $X_k=X_k'$ and for each $i\in H$, $i\in H_j$ for $k<j\leq n$ and $i\in H_j'$ for $k<j\leq m$. Then $H\subseteq H_j$ holds for all $k<j\leq n$ and $H\subseteq H_j'$ holds for all $k<j\leq m$. Since $\K_H \varphi \in ed(s)$, by Proposition \ref{Kntok}, $\K_H\phi \in X_k = X_k'$. By Proposition \ref{Kkton}, $\phi \in X_m' = ed(s')$.
\end{proof}



\begin{lemma}\label{KExist}
If $\K_H \varphi \notin ed(s)$, then there exists a state $s^{\prime} \in S^c$ such that $s \sim^c_H s^{\prime}$ and $\varphi \notin ed\left(s^{\prime}\right)$.
\end{lemma}

\begin{proof}
Let $Y_0 = \{\psi \mid \K_H\psi \in ed(s)\}$. First, we prove that $Y = Y_0 \cup \{\neg\phi\}$ is consistent. Suppose not. Then there exist $\psi_1, \dots, \psi_n \in Y_0$ with $\K_H\psi_1, \dots, \K_H\psi_n \in ed(s)$ such that $$\vdash (\psi_1 \land \dots \land \psi_n) \to \phi.$$ By NEC and Axiom DISTK, we get:
$$
\vdash (\K_H\psi_1 \land \dots \land \K_H\psi_n) \to \K_H\phi.
$$
Since $\K_H\psi_1, \dots, \K_H\psi_n \in ed(s)$, it follows that $\K_H\phi \in ed(s)$, contradicting the assumption $\K_H\phi \notin ed(s)$. Thus $Y$ is consistent. By Lindenbaum's Lemma, there exists a maximal consistent set $Y^*$ such that $Y\subseteq Y^*$. Let $s' = s\la (\top,H),H\ra Y^*$. To verify $s' \in S^c$, we only need to check that $\K_H\psi \in ed(s)$ implies $\psi \in Y^*$ and that $\Kh_H\top\in ed(s)$ and $\K_H\top\in Y^*$, which is guaranteed by the construction of $Y$ and the maximality of $ed(s)$ and $Y^*$. By definition, $s \sim^c_H s'$, making $s'$ the desired state.
\end{proof}




The following lemma is a crucial observation, which is in effect the existence lemma for $\Kh_G$. Note that we will make use of the trivial epistemic relation $\sim_\emptyset$ as mentioned earlier.

\begin{lemma}\label{LemExistence}
    Let $s$ be a state in $S^c$, $G,H$ be groups such that $\emptyset\ne H\subseteq G$, $(\psi,H)$ be an action executable on $[s]_G$. If $\Kh_G\phi\in ed(s')$ for any $s'$ such that $[s]_G\xrightarrow{(\psi,H)}[s']_G$, then $\Kh_G\phi\in ed(s)$.
\end{lemma}

\begin{proof}
    Let $Y_0=\{\chi\mid\K_\emptyset\chi\in ed(s)\}$. We first show that $Y=Y_0\cup\{\K_H\psi,\neg\Kh_G\phi\}$ is inconsistent. Suppose towards a contradiction that $Y$ is consistent, then by Lindenbaum's Lemma, there is an MCS $Y^*$ such that $Y\subseteq Y^*$. Let $s'=s\la(\psi,H),\emptyset\ra Y^*$. By the definition of $Y_0$, for any $\K_\emptyset\chi\in ed(s)$, we have $\chi\in ed(s')$. To verify $s'\in S^c$, we still need to show that $\Kh_H\psi\in ed(s)$ and $\K_H\psi\in Y^*$. Since $\K_H\psi\in Y\subseteq Y^*$, we only need to show that $\Kh_H\psi\in ed(s)$. Since $(\psi, H)$ is executable on $[s]_G$, there exists $t\in S^c$ such that $s\rel{(\psi,H)}t$, then $t=s\la(\psi,H),H'\ra X$ where $H'\subseteq I$ and $X$ is an MCS. Then $\Kh_H\psi\in ed(s)$ and $\K_H\psi\in X$. Therefore $s'\in S^c$. It follows that $s\xrightarrow{(\psi,H)}s'$, then $[s]_G\xrightarrow{(\psi,H)}[s']_G$, which is contradictory with the assumption that $\Kh_G\phi\in ed(s')$ for any $s'$ such that $[s]_G\xrightarrow{(\psi,H)}[s']_G$. Therefore $Y=Y_0\cup\{\K_H\psi,\neg\Kh_G\phi\}$ is inconsistent. Then there are $\chi_1,\cdots,\chi_n\in Y_0$ such that $\Kh_\emptyset\chi_j\in ed(s)$ for each $1\leq j\leq n$ and $\{\chi_1,\dots,\chi_n,\K_H\psi,\neg\Kh_G\phi\}$ is inconsistent. Then 
    $$\vdash(\chi_1\land\dots\land\chi_n)\to(\K_H\psi\to\Kh_G\phi).$$
    By NECK and Axiom DISTK,
    $$\vdash(\K_\emptyset\chi_1\land\dots\land\K_\emptyset\chi_n)\to\K_\emptyset(\K_H\psi\to\Kh_G\phi).$$
    By Axiom AxEmpMono,
    $$\vdash(\K_\emptyset\chi_1\land\dots\land\K_\emptyset\chi_n)\to\K_\emptyset(\Kh_H\K_H\psi\to\Kh_H\Kh_G\phi).$$
    By Axiom T,
    $$\vdash(\K_\emptyset\chi_1\land\dots\land\K_\emptyset\chi_n)\to(\Kh_H\K_H\psi\to\Kh_H\Kh_G\phi).$$
    By Axiom AxKhtoKhK,
    $$\vdash(\K_\emptyset\chi_1\land\dots\land\K_\emptyset\chi_n)\to(\Kh_H\psi\to\Kh_H\Kh_G\phi).$$
    By Axiom AxKhMono and Axiom AxKhKh,
    $$\vdash(\K_\emptyset\chi_1\land\dots\land\K_\emptyset\chi_n)\to(\Kh_H\psi\to\Kh_G\phi).$$
    Since $\K_\emptyset\chi_1,\dots,\K_\emptyset\chi_n,\Kh_H\psi\in ed(s)$, we have $\Kh_G\phi\in ed(s)$.
\end{proof}

We are now ready to prove the Truth Lemma.

\begin{lemma}[Truth Lemma]\label{lemTL}
    For any $\phi\in\DKH$ and $s\in S^c$, $\M^c(X_0),s\vDash\phi$ iff $\phi\in ed(s)$.
\end{lemma}

\begin{proof}
    We show by induction on $\phi$. We only show the case of $\Kh_G\phi$, the other cases are straightforward.

    \textbf{Right to Left: }Suppose that $\Kh_G\phi\in ed(s)$, we show that $\M^c,s\vDash\Kh_G\phi$. There are two cases: $\K_G\phi\in ed(s)$ or $\K_G\phi\notin ed(s)$. If $\K_G\phi\in ed(s)$, then by Lemma \ref{KShare}, $\phi\in ed(s')$ for each $s'\in[s]_G$. By IH, $\M^c,s'\vDash\phi$ for each $s'\in[s]_G$. Therefore, $\M^c,s\vDash\K_G\phi$. By Axiom AxKtoKh and soundness, we have $\M^c,s\vDash\Kh_G\phi$.

    If $\K_G\phi\notin ed(s)$, then $G\ne\emptyset$, otherwise $\K_\emptyset\phi\in ed(s)$ by Axiom AxEmpKhtoK. Consider partial function $\sigma_G=\{[s]_G\mapsto (\phi,G)\}$. We need to show that $(\phi,G)$ is executable on $[s]_G$. Let $v$ be a state in $[s]_G$. Since $\Kh_G\phi\in ed(s)$, then $\K_G\Kh_G\phi\in ed(s)$ by Axiom AxKhtoKKh. By Lemma \ref{KShare} and Lemma \ref{KExist}, $\K_G\Kh_G\phi\in ed(v)$. By Axiom T, $\Kh_G\phi\in ed(v)$. Let $Y_0=\{\chi\mid\K_\emptyset\chi\in ed(v)\}$. We need to show that $Y=Y_0\cup\{\K_G\phi\}$ is consistent. Suppose not, then there are $\chi_1,\dots,\chi_n\in Y_0$ such that $\Kh_\emptyset\chi_j\in ed(v)$ for each $1\leq j\leq n$ and $\{\chi_1,\dots,\chi_n,\K_G\phi\}$ is inconsistent. Then 
    $$\vdash(\chi_1\land\dots\land\chi_n)\to(\K_G\phi\to\bot).$$
By NECK and Axiom DISTK,

    $$\vdash(\K_\emptyset\chi_1\land\dots\land\K_\emptyset\chi_n)\to\K_\emptyset(\K_G\phi\to\bot).$$
By Axiom AxEmpMono,

    $$\vdash(\K_\emptyset\chi_1\land\dots\land\K_\emptyset\chi_n)\to\K_\emptyset(\Kh_G\K_G\phi\to\Kh_G\bot).$$
    By Axiom T,

    $$\vdash(\K_\emptyset\chi_1\land\dots\land\K_\emptyset\chi_n)\to(\Kh_G\K_G\phi\to\Kh_G\bot).$$
    By Axiom AxKhtoKhK,

    $$\vdash(\K_\emptyset\chi_1\land\dots\land\K_\emptyset\chi_n)\to(\Kh_G\phi\to\Kh_G\bot).$$
    By Axiom AxKhbot,

$$\vdash(\K_\emptyset\chi_1\land\dots\land\K_\emptyset\chi_n)\to(\Kh_G\phi\to\bot).$$

Since $\K_\emptyset\chi_1,\dots,\K_\emptyset\chi_n,\Kh_G\phi\in ed(v)$, we have $\bot\in ed(v)$, which is in contradiction with the consistency of $ed(v)$. Therefore, $Y$ is consistent.

By Lindenbaum's Lemma, there is an MCS $Y^*$ such that $Y\subseteq Y^*$. Let $s'=v\la(\phi,G),\emptyset\ra Y^*$. By the definition of $Y$, $\Kh_G\phi\in ed(v)$, $\K_G\phi\in ed(s')$ and $\chi\in ed(s')$ for any $\K_\emptyset\chi\in ed(v)$, therefore $s'\in S^c$. Then $v\xrightarrow{(\phi,G)}s'$. Therefore $(\phi,G)$ is executable on $[s]_G$. Let $[s'']_G$ be an equivalence class such that $[s]_G\xrightarrow{(\phi,G)}[s'']_G$, then there are $t'\in[s]_G$ and $t''\in[s'']_G$ such that $t'\rel{(\phi,G)}t''$. It follows that $t''=t'\la(\phi,G),H\ra X$ where $H\subseteq I$ and $X$ is an MCS. Then $\Kh_G\phi\in ed(t')$ and $\K_G\phi\in ed(t'')$. By Lemma \ref{KShare} and Lemma \ref{KExist}, $\K_G\phi,\phi\in ed(t)$ for each $t\in [s'']_G$. By IH, $\M^c(X_0),t\vDash\phi$. Moreover, since $\neg\K_G\phi\in ed(s)$, it is not the case that $[s]_G\xrightarrow{(\phi,G)}[s]_G$, then all complete executions of $\sigma_G=\{[s]_G\mapsto (\phi,G)\}$ starting from $[s]_G$ are finite. Therefore, $\M^c(X_0),s\vDash\Kh_G\phi$.

\textbf{Left to Right: }Suppose that $\M^c(X_0),s\vDash\Kh_G\phi$, we show that $\Kh_G\phi\in ed(s)$. By semantics, there is a strategy $\sigma_G$ such that
\begin{enumerate}
\item $[t]_G\subseteq\llrr{\varphi}$ for all $[t]_G\in\CELeafN(\sigma_G,s)$, and 
\item all its complete executions starting from $[s]_G$ are finite.
\end{enumerate}
For each $t'\in [t]_G$, by IH, $\phi\in ed(t')$. By Lemma \ref{KExist}, $\K_G\phi\in ed(t)$. By Axiom AxKtoKh, $\Kh_G\phi\in ed(t)$.

If $[s]_G\notin\Dom(\sigma_G)$, then $[s]_G\in\CELeafN(\sigma_G,s)$, then $\Kh_G\phi\in ed(s)$. If $[s]_G\in\Dom(\sigma_G)$, then $G\ne\emptyset$ and $[s]_G\in\CEInnerN(\sigma_G,s)$. In order to show that $\Kh_G\phi\in ed(s)$, we will show a stronger result, that is, $\Kh_G\phi\in ed(s')$ for all $[s']_G\in\CEInnerN(\sigma_G,s)$. We firstly show the following claim:

\begin{claim}\label{infinite}
    If there exists $[s']_G\in\CEInnerN(\sigma_G,s)$ such that $\neg\Kh_G\phi\in ed(s')$, then there exists an infinite execution of $\sigma_G$ starting from $[s]_G$.
\end{claim}

\begin{claimproof}
    Suppose that there exists $[s']_G\in\CEInnerN(\sigma_G,s)$ such that $\neg\Kh_G\phi\in ed(s')$. Let $X=\{[v]_G\in\CEInnerN(\sigma_G,s)\mid\neg\Kh_G\phi\in ed(v)\}$, then $[s']_G\in X$, and for any $[v]_G\in X$, since $[v]_G\in\CEInnerN(\sigma_G,s)$, we have $[v]_G\in\Dom(\sigma_G)$, then $\sigma_G([v]_G)$ is executable on $[v]_G$. By Proposition \ref{PropIntNonexec}, $\sigma_G([v]_G)\in A_G^{c+}$, so it has the form of $(\psi,H)$ where $\emptyset\ne H\subseteq G$.
    
    Define binary relation on $[S^c]_G$ as $R=\{([v]_G,[v']_G)\mid [v]_G\xrightarrow{\sigma_G([v]_G)}[v']_G\}$. For any $[v]_G\in X$, since $\neg\Kh_G\phi\in ed(v)$, by Lemma \ref{LemExistence}, there exists $v'\in S^c$ such that $[v]_G\xrightarrow{\sigma_G([v]_G)}[v']_G$ and $\neg\Kh_G\phi\in ed(v')$. Since $[v]_G\in X$, then $[v]_G\in\CEInnerN(\sigma_G,s)$. Since $\Kh_G\phi\in ed(t)$ for all $[t]_G\in\CELeafN(\sigma_G,s)$, then $[v']_G\notin\CELeafN(\sigma_G,s)$. Since $[v]_G\xrightarrow{\sigma_G([v]_G)}[v']_G$, then $[v']_G\in\CEInnerN(\sigma_G,s)$, then $[v']_G\in X$, then $([v]_G,[v']_G)\in R$. Therefore, $R$ is an entire binary relation on $X$, that is, for each $[v]_G\in X$, there is $[v']_G\in X$ such that $([v]_G,[v']_G)\in R$. By Axiom of Dependent Choice, there exists an infinite sequence $[v_0]_G[v_1]_G\cdots$ such that $([v_n]_G,[v_{n+1}]_G)\in R$ for all $n\in\mathbb{N}$.
    
    By definition of $R$, $[v_0]_G[v_1]_G\cdots$ is a complete execution of $\sigma_G$ starting from $[v_0]$. Since all complete execution of $\sigma_G$ starting from $[s]_G$ are finite and $[v_0]_G\in\CEInnerN(\sigma_G,s)$, there exists an execution of $\sigma_G$ $[s_0]_G\cdots [s_j]_G$ $(j\in\mathbb{N})$ such that $[s_0]_G=[s]_G$ and $[s_j]_G=[v_0]_G$. Then $[s_0]_G\cdots [s_j]_G[v_1]_G\cdots$ is an infinite execution of $\sigma_G$ starting from $[s]_G$.
    \end{claimproof}
    
    For any $[s']_G\in\CEInnerN(\sigma_G,s)$, suppose that $\neg\Kh_G\phi\in ed(s')$. By Claim \ref{infinite}, there exists an infinite execution of $\sigma_G$ starting from $[s]_G$, which is in contradiction with all complete execution of $\sigma_G$ starting from $[s]_G$ are finite. Therefore, $\Kh_G\phi\in ed(s')$ for all $[s']_G\in\CEInnerN(\sigma_G,s)$. Since $[s]_G\in\CEInnerN(\sigma_G,s)$, we have $\Kh_G\phi\in ed(s)$.
\end{proof}

\begin{theorem}
    $\SDKh$ is strongly complete.
\end{theorem}

\section{Conclusions and future work} \label{sec.conc}

In this paper, we propose a framework of group knowledge-how, featuring the distributed actions of groups that can come from three sources. The framework can be viewed as a generalization of both the planning-based approach and the coalition-based approach to the logic of knowing how. We give a sound and complete proof system of the logic with intuitive axioms. We leave the discussions on the decidability and model-theoretical properties to a future occasion. 

There are many further directions to explore based on this very general framework. One interesting direction is to consider subclasses of the models and study their corresponding logics. For example, one may want to consider the case when the atomic group actions are all empty except the singleton ones, which will bring our framework closer to the setting of coalition logic and (epistemic) ATL such as \cite{naumov2017together,naumov2018together,vanderHoek2003,JamrogaA07}. Conceptually, this also makes sense if one really wants the distributed knowledge-how to be always decomposable into individual know-how. It will also help us to better understand the differences and connections with other related existing frameworks. 

One may notice that in our current setting, a joint action consisting of multiple subgroups doing the same action is equivalent to a single subgroup doing it, \footnote{It is similar in the setting of STIT logic, where ``actions'' are formalized as possible outcomes \cite{Broersen2015}.} e.g., the transition relation of $\lr{a,a}$ is the intersection of $\rel{a}$ and itself, which is exactly $\rel{a}$. This is different from the ATL-like frameworks, where the effect of the joint action $\lr{a,a}$ can be different from $a$ \cite{jamroga2004agents,Bulling2015}. This invites us to extend our work in a non-trivial way to handle resource-sensitive settings where the number of individuals performing the same action matters. We may also add more structure to groups or consider intensional groups as discussed in \cite{MartaIgor23tark}. Also, the members of the groups may not be equal in their roles as in \cite{dunin2004tuning}. In particular, an interesting direction to explore is the possibility that each group has a leader who can plan the actions of the members. The commitments and communication among the leaders may help the groups synchronize better. 

Given our discussions on extra group actions which cannot be reduced to the individual actions, it may be interesting to find out whether this can be applied to distributed knowledge-that as well: the epistemic relation $\sim_G$ is not defined by the intersection of the individual ones but just a subset of it. More precisely, in addition to distributed knowledge where $\sim_G=\bigcap_{i\in G}\sim_i$, through brainstorming, a group can be smarter than the mere union of individuals, so $\bigcap_{i\in G}\sim_i\subseteq\sim_G$ does not necessarily hold. However, we think the logic stays the same. 

 Finally, inspired by various versions of distributed knowledge-how, e.g., \cite{GalimullinK24,BalbianiD24}, we may also consider other definitions of distributed knowledge-how that incorporate dynamics as well.

\paragraph{Acknowledgment} The authors thank Yanjun Li and Pavel Naumov for their comments and discussions about earlier versions of the paper and related topics. The authors are also grateful to the anonymous reviewers, whose comments improved the presentation of the paper. 


\bibliographystyle{eptcs}
\bibliography{generic}

\end{document}